\documentclass[a4paper,11pt]{article}
\usepackage{etoolbox}
\newtoggle{anonymous}
\togglefalse{anonymous}

\usepackage{fullpage}
\usepackage{amsmath,amssymb,amsfonts}%
\usepackage{amsthm}%
\usepackage{mathrsfs}%
\usepackage{graphicx}%
\usepackage{multirow}%

\usepackage{multirow}%

\usepackage{dsfont}
\usepackage{mathrsfs}
\usepackage{amssymb,amsmath,amsfonts,amsthm}
\usepackage[parfill]{parskip}
\usepackage{enumitem} % for custom enumeration
\setlist[enumerate,1]{label={(\roman*)}}
\usepackage{url} % for URL link
\usepackage{xspace}
\usepackage[most]{tcolorbox}

\usepackage[ruled,vlined]{algorithm2e}

\usepackage{changepage}
\usepackage{thmtools}

\usepackage{graphicx}
\usepackage{subcaption}

\usepackage{tikz}
\usetikzlibrary{calc}
\usetikzlibrary{positioning, backgrounds, fit}
\usetikzlibrary{decorations.pathmorphing}
\usetikzlibrary{intersections}

\usepackage{xcolor} % For custom colors

\definecolor{FigRed}{RGB}{046,037,133}
\definecolor{FigBlue}{RGB}{051,117,056}
\definecolor{FigGreen}{RGB}{148,203,236}

\definecolor{linkGreen}{RGB}{175,0,0}

\definecolor{linkblue}{RGB}{0, 102, 204} 
\usepackage[
colorlinks=true,
linkcolor=linkblue,     % Color for links (e.g., theorems, equations)
citecolor=linkGreen,     % Color for citations
urlcolor=linkblue       % Color for URLs
]{hyperref}

\usepackage{dirtytalk} % quotations with \say
\usepackage{mathtools} % to use \coloneqq

\usepackage[capitalize,nameinlink]{cleveref}

\Crefname{figure}{Figure}{Figures}
\Crefname{claim}{Claim}{Claims}
\Crefname{algorithm}{Algorithm}{Algorithms}

\crefformat{equation}{\textup{#2(#1)#3}}
\crefrangeformat{equation}{\textup{#3(#1)#4--#5(#2)#6}}
\crefmultiformat{equation}{\textup{#2(#1)#3}}{ and \textup{#2(#1)#3}}
{, \textup{#2(#1)#3}}{, and \textup{#2(#1)#3}}
\crefrangemultiformat{equation}{\textup{#3(#1)#4--#5(#2)#6}}%
{ and \textup{#3(#1)#4--#5(#2)#6}}{, \textup{#3(#1)#4--#5(#2)#6}}{, and \textup{#3(#1)#4--#5(#2)#6}}

\Crefformat{equation}{Equation~\textup{#2(#1)#3}}
\Crefrangeformat{equation}{Equations~\textup{#3(#1)#4--#5(#2)#6}}
\Crefmultiformat{equation}{Equations~\textup{#2(#1)#3}}{ and \textup{#2(#1)#3}}
{, \textup{#2(#1)#3}}{, and \textup{#2(#1)#3}}
\Crefrangemultiformat{equation}{Equations~\textup{#3(#1)#4--#5(#2)#6}}%
{ and \textup{#3(#1)#4--#5(#2)#6}}{, \textup{#3(#1)#4--#5(#2)#6}}{, and \textup{#3(#1)#4--#5(#2)#6}}

\newtheorem{theorem}{Theorem}[section]

\newtheorem{corollary}[theorem]{Corollary}
\newtheorem{observation}[theorem]{Observation}

\newtheorem{lemma}[theorem]{Lemma}

\theoremstyle{definition}

\theoremstyle{remark}

\newcommand{\cG}{\ensuremath{\mathcal{G}}}

\newcommand{\bfI}{\ensuremath{\mathcal{I}}}

\usepackage{soul}
\newcommand{\SuggestChange}[2]{{\color{red} \relax\ifmmode\text{\st{$#1$}}\else \st{#1}\fi}\ {\color{blue} #2}}

\definecolor{samplingboxcolor}{HTML}{FCF5E5} % Light pastel beige
\definecolor{titleboxcolor}{HTML}{EAEAEA}   % Light gray for title
\definecolor{bordercolor}{rgb}{0.7, 0.7, 0.7} % Soft gray for border

\newtcolorbox[auto counter, number within=section]{boxsampling}[1][]{
	colback=samplingboxcolor,  % Light background color
	colframe=bordercolor,      % Soft gray border
	coltitle=black,            % Black title text for readability
	colbacktitle=titleboxcolor,% Light gray title background
	fonttitle=\bfseries\large, % Larger and bold title font
	title=#1,                  % Uses the argument as the title
	rounded corners,           % Softer edges
	boxrule=0.75pt,            % Thin border for subtle look
	enhanced                   % Enables advanced styling
}

\newtcolorbox[auto counter, number within=section]{boxproblem}[1][]{
	colback=samplingboxcolor,  % Light background color
	colframe=bordercolor,      % Soft gray border
	coltitle=black,            % Black title text for readability
	colbacktitle=titleboxcolor,% Light gray title background
	fonttitle=\scshape\large, % Larger and bold title font
	title=#1,                  % Uses the argument as the title
	rounded corners,           % Softer edges
	boxrule=0.75pt,            % Thin border for subtle look
	enhanced                   % Enables advanced styling
}

\usepackage{thm-restate}

\usepackage{enumitem,amssymb}
\newlist{todolist}{itemize}{2}
\setlist[todolist]{label=$\square$}
\usepackage{pifont}

\allowdisplaybreaks

\newcommand{\cint}[2]{[\,#1,\,#2\,]_{\circlearrowright}}
\newcommand{\coint}[2]{[\,#1,\,#2\,)_{\circlearrowright}}

\usepackage{thm-restate}

\allowdisplaybreaks

\title{Exploration of $k$-edge-deficient temporal graphs in linear time}

\iftoggle{anonymous}{% ANONYMOUS
	\author{ }
}{%
	\author{Ivan Lahtin\thanks{Moscow Institute of Physics and Technology (MIPT), Russia} 
	\and
	Viktor Zamaraev\thanks{School of Computer Science and Informatics, University of Liverpool, UK}
	}
}
\date{}

\usepackage{url}

\begin{document}
	
	\maketitle

    \begin{abstract}
        We study the \emph{Temporal Exploration} problem, where an agent must visit all vertices of a temporal graph while traversing at most one available edge per time step. Unlike static graphs, which can be explored in linear time, temporal constraints can substantially increase exploration time even when every snapshot of the graph is connected.

        To better understand the source of this complexity, we focus on a near-static setting and consider \emph{always-connected $k$-edge-deficient temporal graphs}, in which each snapshot is connected and differs from a fixed underlying $n$-vertex graph by at most $k$ edges. Although such graphs are structurally close to static graphs, they can still exhibit non-trivial temporal behaviour. Prior work showed that these graphs can be explored in $O(kn \log n)$ time steps and established a lower bound of $\Omega(n \log k)$, leaving open whether linear-time exploration in $n$ is possible.

        We resolve this question by showing that any always-connected $k$-edge-deficient temporal graph admits an exploration schedule of length $O(nk \log k)$. Moreover, given such a temporal graph, the corresponding exploration schedule can be computed in polynomial time.
        The obtained bound is linear in the number of vertices up to a factor depending only on $k$, removes the extraneous logarithmic dependence on $n$, and is nearly optimal. In particular, for constant $k$, our result yields an order-optimal $\Theta(n)$ exploration time, showing that temporal exploration in this near-static regime essentially retains the linear-time character of static graph traversal.
    \end{abstract}

	\newpage

	%%%%%%%%%%%%%%%%%%%%%%%%%%%%%
	\section{Introduction}
	\label{sec:intro}
	%%%%%%%%%%%%%%%%%%%%%%%%%%%%% 

	Temporal graphs—graphs whose edge sets evolve over time—provide a natural abstraction for dynamic networks in which interactions are inherently time-dependent, and have been studied extensively in the literature (see surveys~\cite{CFQS12,MichailSurvey}). Such models arise in a wide range of settings, including mobile communication networks, transportation systems, and evolving social or biological networks (see, e.g.,~\cite{HS12,HS19}). In contrast to static graphs, temporal graphs capture not only which connections are possible, but also \emph{when} they are available, and this additional temporal dimension fundamentally alters the algorithmic behaviour of even the most basic graph problems.
	
	One of the most elementary algorithmic tasks in this setting is \emph{temporal exploration}. In the \emph{Temporal Exploration} problem, introduced by Michail and Spirakis~\cite{MS16}, an agent starts at a designated vertex and must visit all vertices of a temporal graph while traversing at most one available edge per time step. Temporal exploration abstracts navigation and information gathering in dynamic environments and is commonly used to illustrate how temporal constraints can fundamentally change the complexity of otherwise simple graph problems. In static graphs, exploration is straightforward: any connected graph can be traversed in linear time using depth-first search. A natural question is whether comparable efficiency can be achieved when edges change over time.
	
    The answer, in general, is negative. Temporal constraints can dramatically increase the cost of exploration. 
    In fact, there exist temporal graphs in which every snapshot is connected, yet any exploration requires quadratic time in the number of vertices~\cite{EHK21}. This stark contrast with the static setting highlights that temporal exploration is not merely a minor extension of classical graph traversal, but a fundamentally different algorithmic problem in which time itself becomes a limiting resource.
	
	This intuition is reinforced by a range of strong hardness results. Beyond large worst-case exploration times, even deciding whether a temporal exploration exists within a prescribed lifetime is NP-complete~\cite{MS16}, and this intractability persists even under extreme structural restrictions, including instances where each snapshot is a tree and the underlying graph has pathwidth~$2$~\cite{BZ19}.
	These results demonstrate that restricting the structure of the underlying graph alone does not suffice to recover tractability, and that the difficulty of temporal exploration cannot be attributed solely to static graph complexity.

	To isolate the source of this difficulty, much recent work has focused on \emph{always-connected temporal graphs}, in which every snapshot is connected. This restriction eliminates disconnection as a trivial obstruction and ensures that some route exists between any pair of vertices within any $n$ consecutive steps \cite{MS16}. Thus, always-connected temporal graphs provide a clean and well-motivated framework in which temporal effects are the main source of inefficiency.
	
    A general upper bound of $O(n^2)$ steps is known for exploration in always-connected temporal graphs, and a substantial body of work has aimed to improve this bound under additional restrictions. 
    Faster exploration has been obtained for various special classes, including temporal graphs of bounded degree~\cite{ES18,EKLSS19,BGMR25}, temporal graphs with additional symmetries~\cite{DEKMM24}, graphs with specific underlying structures~\cite{IKW14,EHK21,Tag20,AGMZ22}, as well as in randomized~\cite{EHK21,BGKSSZ26} and parameterised~\cite{ES23,AFGW23} settings. Despite this progress, significant gaps between known upper and lower bounds persist even for some fundamental cases. Consequently, we still lack a clear structural understanding of when temporal exploration is inherently hard and when it can approach the linear-time behaviour of static graphs.

	To better understand which aspects of temporal variation are responsible for this gap, a particularly appealing direction is to study temporal graphs that remain close to a fixed static structure over time. 
    One formalization of this intuition is \emph{$k$-edge-deficiency}, introduced by Erlebach and Spooner~\cite{ES22}, where each snapshot differs from a fixed underlying graph by the absence of at most $k$ edges. Such graphs can be viewed as near-static: temporally, they undergo only limited disruption, yet they are not static and may still exhibit complex temporal behaviour. This setting offers a natural testbed for understanding how small temporal deviations impact exploration complexity.

    For certain highly structured temporal graphs, $k$-edge-deficiency does permit linear-time exploration. In particular, always-connected temporal graphs whose underlying graph is a cycle or a cycle with a constant number of chords can be explored in a linear number of steps~\cite{EHK21,Tag20,AGMZ22}. However, these results rely crucially on strong restrictions on the underlying graph and do not extend to general $k$-edge-deficient temporal graphs, where prior to this work only superlinear upper bounds were known.
	
	In particular, Erlebach and Spooner~\cite{ES22} showed that always-connected $k$-edge-deficient temporal graphs admit exploration schedules of length $O(kn \log n)$, and that $\Omega(n \log k)$ steps may be necessary in the worst case. These results establish that $k$-edge-deficiency does not trivialize temporal exploration, but they leave open a central question: is the logarithmic dependence on $n$ inherent, or can exploration be made linear in the number of vertices, up to a factor depending only on $k$? 
	
	In the present paper, we answer this question in the affirmative. We show that any always-connected $k$-edge-deficient temporal graph on $n$ vertices admits an exploration schedule of length $O(n k \log k)$. Moreover, given such a temporal graph, the corresponding schedule can be computed by a Las Vegas\footnote{A
    Las Vegas randomized algorithm is a randomized algorithm that always returns a correct solution.} algorithm in expected polynomial time, and deterministically in polynomial time when $k$ is constant. The above bound is linear in $n$ for fixed $k$, removes the $\log n$ dependence 
    and nearly matches the known lower bound up to a factor of $k$.

    %%%%%%%%%%%%%%%%%%%%%%%%%%%%%%%%%%%%%%%%%%%%%%
	\subsection{Our contribution}
    %%%%%%%%%%%%%%%%%%%%%%%%%%%%%%%%%%%%%%%%%%%%%%
    
    Informally, our main contribution has two parts. First, we prove that if a temporal graph remains temporally connected over every window of $\Delta$ consecutive time steps and admits a fixed spanning tree such that in every snapshot at most $k$ edges of this tree are absent, then it can be explored in $O((k\Delta + n)\log k)$ steps. 
    Second, we show that given such a graph, the corresponding exploration schedule can be computed efficiently.
    To state these results formally, we introduce the necessary definitions and notation,
    which slightly generalize the standard notions of always-connectedness and
    $k$-edge-deficiency used in previous work.
    
    A \emph{temporal graph $\cG$} is an ordered sequence
    $\langle G_1,G_2,\dots,G_L\rangle$ of graphs on the same vertex set.
    The graphs in the sequence are called the \emph{snapshots} of $\cG$, and the number
    $L$ of these graphs is called the \emph{lifetime} of $\cG$.
    The \emph{underlying graph} of $\cG$ is the (static) graph with the same vertex set
    as $\cG$ that contains exactly those edges that appear in at least one snapshot of
    $\cG$.
    For integers $1 \le t \le t' \le L$, we denote by $\cG_{[t,t']}$ the temporal graph
    $\langle G_t, G_{t+1}, \dots, G_{t'} \rangle$, that is, the restriction of $\cG$ to
    the consecutive snapshots with indices from $t$ to $t'$.
    A \emph{temporal walk} in $\cG$ from a vertex $v$ to a vertex $u$ is a sequence
    $(v = v_0, e_1, v_1, e_2, \dots, e_\ell, v_\ell = u)$ together with a strictly
    increasing sequence of time steps
    $1 \le t_1 < t_2 < \dots < t_\ell \le L$
    such that for each $i \in [\ell]$, the edge
    $e_i = \{v_{i-1}, v_i\}$ is present in snapshot $G_{t_i}$.
    The \emph{length} of such a temporal walk is the number $\ell$ of its edges.
    We say that $\cG$ can be \emph{explored from vertex $v$ in at most $\ell$ steps} if there exists a temporal walk of length at most $\ell$ that starts at $v$ and visits
    all vertices of $\cG$.
    We say that $\cG$ can be \emph{explored from vertex $v$} if it can be explored from $v$ within its lifetime, i.e., in at most $L$ steps.
    
    Temporal graph $\cG$ is \emph{temporally connected} if for every ordered pair of
    vertices $v,u$, there exists a temporal walk from $v$ to $u$ in $\cG$.
    For a natural number $\Delta$, we say that $\cG$ is \emph{$\Delta$-temporally
    connected} if for every $t \in [L-\Delta+1]$, the temporal graph
    $\cG_{[t, t + \Delta - 1]}$ is temporally connected.

    For a natural number $k$ and a (static) graph $F$ on the same vertex set as $\cG$, we say that a snapshot $G$ of $\cG$ is \emph{$k$-edge-deficient with respect to $F$} if all but at most $k$ edges of $F$ are present in $G$. When $F$ is clear from the context we may omit "with respect to $F$". We say that $\cG$ is \emph{$k$-edge-deficient} if all snapshots of $\cG$ are $k$-edge-deficient with respect to a spanning tree of the underlying graph of $\cG$.
    
    These definitions allow us to state our result in a more general setting than the
    always-connected $k$-edge-deficient temporal graphs considered previously.
    Instead of assuming that every snapshot is connected, we work with the more
    flexible notion of $\Delta$-temporal connectivity.
    The parameter $\Delta$ is introduced to measure the cost of temporal repositioning separately from the cost of traversing the near-static spanning tree.
    Every always-connected temporal graph on $n$ vertices is
    $(n-1)$-temporally connected~\cite{MS16}, whereas the converse does not hold in
    general, and the dependence of our bound on $\Delta$ explicitly captures this
    relaxation.
    Moreover, our notion of $k$-edge-deficiency is weaker than that of Erlebach and
    Spooner~\cite{ES22}, who require each snapshot to differ from the \emph{underlying graph} by at most $k$ edges; in contrast, we only assume the existence of a spanning tree whose edges are stable in the sense that in every snapshot all but at most $k$ edges of this tree are present.
    This permits substantially more temporal variability outside the tree structure,
    while still retaining a near-static backbone that can be exploited for fast exploration.
    With the above terminology at hand, we can now state our main results.

    \begin{restatable}{theorem}{KDefExploration}
        \label{th:k-def-exploration}
        Let $\Delta$, $k$, and $n$ be natural numbers. Let $\cG$ be a $\Delta$-temporally connected temporal graph on $n$ vertices, and let $T$ be a spanning tree of the underlying graph of $\cG$. If $\cG$ contains at least $\lceil 18k \ln{(6k)} \rceil(\Delta+ \frac{n}{k})$ snapshots that are $k$-edge-deficient with respect to $T$, then $\cG$ can be explored from any vertex.
    \end{restatable}

    For always-connected $k$-edge-deficient temporal graphs, where $\Delta = n-1$ and every snapshot is $k$-edge-deficient with respect to some fixed tree, \cref{th:k-def-exploration} yields an
    exploration time of $O(nk \log k)$, improving upon the earlier best-known $O(kn \log n)$ bound~\cite{ES22} by eliminating the dependence on $\log n$. More generally, the bound scales linearly with the temporal connectivity parameter $\Delta$, revealing that faster exploration is possible when connectivity is guaranteed over shorter time windows. In particular, when $\Delta$ is constant we obtain an $O(n \log k)$ exploration time, which matches the $\Omega(n \log k)$ lower bound of Erlebach and Spooner~\cite{ES22}.

    We note that \cref{th:k-def-exploration} does not require every snapshot to be $k$-edge-deficient, and therefore applies to a more general class of temporal graphs than standard $k$-edge-deficient temporal graphs. We use this additional flexibility to derive an efficient algorithm for computing exploration schedules in the $k$-edge-deficient setting.

    \begin{restatable}{theorem}{KDefExplorationAlg}
        \label{th:k-def-exploration-alg}
        Let $\Delta$, $k$, and $n$ be natural numbers.
        Let $\cG$ be a $\Delta$-temporally connected $k$-edge-deficient temporal graph on $n$ vertices. If the lifetime of $\cG$ is at least
        $2\lceil 36k \ln{(12k)} \rceil(\Delta+ \frac{n}{2k})$,
        then $\cG$ can be explored from any vertex. Moreover, given $\cG$, the corresponding exploration schedule can be computed by a Las Vegas algorithm in expected polynomial time, and deterministically in polynomial time when $k$ is constant.
    \end{restatable}

    Specializing \cref{th:k-def-exploration-alg} to always-connected temporal graphs gives the headline bound, together with the corresponding algorithmic guarantee.

    \begin{corollary}\label{cor:k-deficient-always-connected}
        Let $k$ and $n$ be natural numbers. Let $\cG$ be an always-connected $k$-edge-deficient temporal graph on $n$ vertices with sufficiently many snapshots.
        Then $\cG$ can be explored from any vertex in $O(n k \log k)$ steps. Furthermore, given $\cG$, the corresponding exploration schedule can be computed by a Las Vegas algorithm in expected polynomial time, and deterministically in polynomial time when $k$ is constant.
    \end{corollary}

    We make no attempt to optimize the absolute constants in the statements or proofs; they are chosen for convenience to simplify the estimates.
    %%%%%%%%%%%%%%%%%%%%%%%%%%%%%%%%%%%%%%%%%%%%%%
	\subsection{Proof outline}
    %%%%%%%%%%%%%%%%%%%%%%%%%%%%%%%%%%%%%%%%%%%%%%
    
    We briefly outline the main ideas behind the proofs of \cref{th:k-def-exploration} and \cref{th:k-def-exploration-alg}. Full details are deferred to \cref{sec:bound} and \cref{sec:efficient-computation} respectively.

    We start by outlining the proof of \cref{th:k-def-exploration}.
    \begin{enumerate}[label=\arabic*.]
        \item \textbf{Reduction to a DFS tour of a stable backbone.}  
        We fix a spanning tree $T$ from the statement of \cref{th:k-def-exploration}, and a depth-first search tour of $T$. Exploration of $\cG$ is reduced to covering this DFS tour.

        \item \textbf{Parallel exploration via the roundabout process.}
        We consider the \emph{roundabout exploration process}, in which we initially place one virtual agent at each position of the DFS tour.
        During this stage we only use snapshots that are $k$-edge-deficient with respect to $T$, skipping all other snapshots.
        The process proceeds in iterations, each consisting of two phases. In the \emph{movement phase}, agents advance along the DFS tour whenever the corresponding tree edges are available. In the subsequent \emph{elimination phase}, redundant agents, namely those whose explored portions of the tour are already covered by the remaining agents, are removed.
        The central invariant of the roundabout exploration process is that after each iteration, the union of the tour segments visited by the active agents covers the entire DFS tour, and hence all vertices of $\cG$.
        In \cref{sec:roundabout-exploration}, we show that after $O(n/k)$ iterations, the number of active agents drops to $O(k)$.

        \item \textbf{Iteration over multiple epochs and selection of agents.}  
        To convert the parallel exploration performed by multiple virtual agents into an exploration schedule for a single agent, we divide the timeline into consecutive epochs. Each epoch consists of a repositioning part of $\Delta$ time steps, followed by a roundabout part. At the end of each epoch, the roundabout part produces a set of $O(k)$ agents whose explored segments collectively cover the entire DFS tour.
        While no single agent covers the entire DFS tour within a single epoch, we show in \cref{sec:patroling-agents} that across a sequence of sufficiently many epochs, there exists a choice of one agent per epoch whose combined explored segments span the DFS tour.

        \item \textbf{Simulation by a single explorer.}  
        Having identified such agents, we show in \cref{sec:single-agent-exploration} how a single explorer can simulate their movements. 
        The $\Delta$-temporal connectivity property ensures that the explorer can reposition itself during the repositioning part between consecutive roundabout parts, allowing it to follow the exploration schedule of the chosen agents epoch by epoch.

        \item \textbf{Deriving the time bound.}  
        By choosing the number of epochs to be $\lceil 18k \ln{(6k)} \rceil$, we derive that $\lceil 18k \ln{(6k)} \rceil(\Delta+ \frac{n}{k})$ $k$-edge-deficient snapshots is enough to explore $\cG$, completing the proof of \cref{th:k-def-exploration}.
    \end{enumerate}

    Let us briefly contrast this strategy with the earlier $O(kn\log n)$ approach of
    Erlebach and Spooner~\cite{ES22}. Their algorithm also uses virtual agents, but organizes
    them around a balanced cover of the unexplored part of a spanning tree by $k+1$ subtrees.
    In each round, one virtual agent is placed in each subtree and attempts to follow a DFS tour of that subtree. Since at most $k$ edges are missing in any snapshot, some subtree is unblocked often enough for its agent to complete its tour; that subtree is then removed. This yields a geometric decrease in the number of unexplored tree edges and hence $O(k\log n)$ rounds.
    
    Our proof uses virtual agents in a different way. Instead of repeatedly decomposing the
    remaining subtrees, we use a single DFS tour of the entire spanning tree. In each epoch, the roundabout process eliminates redundant agents while preserving coverage of the whole tour, leaving only $O(k)$ representative agents after $O(n/k)$ steps.
    By combining $O(k\log k)$ such epochs, we construct a single-agent exploration schedule and thereby remove the dependence on $\log n$.

    To establish \cref{th:k-def-exploration-alg} we first show in \cref{sec:efficient-computation} that when the spanning tree $T$ is given, the proof of \cref{th:k-def-exploration} can be made algorithmic, i.e., the corresponding exploration schedule can be computed efficiently. 
    Then, in \cref{sec:efficient-computation}, we prove that, given a $k$-edge-deficient temporal graph $\cG$, one can efficiently compute a spanning tree $T$ such that at least half of the snapshots in $\cG$ are $2k$-edge-deficient with respect to $T$. 
    Combining the two components yields an efficient algorithm for computing an exploration schedule for $k$-edge-deficient temporal graphs.

	%%%%%%%%%%%%%%%%%%%%%%%%%%%%%
	\section{Exploration bound}
	\label{sec:bound}
	%%%%%%%%%%%%%%%%%%%%%%%%%%%%%

    This section proves \cref{th:k-def-exploration} by reducing temporal exploration to coverage of a DFS tour of the stable spanning tree $T$. The proof first analyses a multi-agent parallel exploration process on this tour, then repeats this process over several epochs, and finally uses $\Delta$-temporal connectivity to simulate selected virtual agents by a single explorer. We begin by fixing the notation used throughout the proof.

    Let $L$ be a sufficiently large natural number, and
	let $\cG=\langle G_1,G_2,\dots,G_L\rangle$ be a $\Delta$-temporally connected temporal graph on $n$ vertices with underlying
	graph $G=(V,E)$. Let $T$ be a spanning tree of $G$ and assume that $\cG$ contains at least $\lceil 18k \ln{(6k)} \rceil(\Delta+ \frac{n}{k})$ snapshots, that are $k$-edge-deficient with respect to $T$.
    Fix a root $r \in V$ in $T$ and a DFS tour of $T$ that starts and ends at $r$ and traverses
	each edge of $T$ exactly twice. Denote by
	\[
		(v_1, v_2, \ldots, v_N, v_{N+1}) \quad\text{with}\quad v_{N+1}=v_1=r
	\]
	the cyclic sequence of vertices along this tour, where $N:=2(n-1)$ is the number
	of tour edges. 
	Denote the tour edges by $e_i:=\{v_i,v_{i+1}\}$, $i \in N$, and note that each such edge is an edge of $T$.
	For indices $i, j \in [N]$, we define the \emph{circular interval}
	\[
		\cint{i}{j} :=
		\begin{cases}
			\{i, i+1, \dots, j\}, & \text{if } i \le j,\\[4pt]
			\{i, i+1, \dots, N, 1, 2, \dots, j\}, & \text{if } i > j.
		\end{cases}
	\]
	That is, $\cint{i}{j}$ consists of all indices encountered when moving forward from $i$ to $j$ along the cycle $(1,2, \ldots, N, 1)$. Similarly, we denote by $\coint{i}{j}$ the circular interval obtained from $\cint{i}{j}$ by excluding its right endpoint $j$.

	For each $i \in [N]$, we have an \emph{agent} $a_i$ whose \emph{state} at any time is an element of $[N]$, which reflects the current position of the agent along the DFS tour, i.e., if the state of an agent is equal to $j$, then the agent is positioned at vertex $v_j$. 
	We denote $A := \{a_1, a_2, \ldots, a_N \}$.
	
	%%%%%%%%%%%%%%%%%%%%%%%%%%%%	
	\subsection{Roundabout exploration process}\label{sec:roundabout-exploration}
	%%%%%%%%%%%%%%%%%%%%%%%%%%%%	
    The purpose of the \emph{roundabout exploration} process is to identify a small number of agents that explore the DFS tour in a short time interval.
    In the beginning, the $N$ agents start moving around the DFS tour like runners on a circular track. Although some agents may be blocked by missing tree edges, redundant runners are removed once their covered intervals are already covered by others.
    The main result of this subsection is that, after $\lfloor N/2k\rfloor$ steps, only
    $O(k)$ active agents remain.
    
    The roundabout exploration process can be executed on any sequence of at most $N$ snapshots each of which is
    $k$-edge-deficient with respect to $T$. Thus, for convenience, in this subsection we consider
    $\langle G_1,G_2,\ldots,G_N\rangle$ and assume that all snapshots in this temporal graph are
    $k$-edge-deficient with respect to $T$. The process runs for steps
    $t=1,2,\ldots,N$, and each step consists of the following two phases:
	\begin{enumerate}
		\item \textbf{Movement}. At step $t$, each agent $a_i$ attempts to move from its current state $j$ to the next state $j+1$ (i.e., from $v_j$ to $v_{j+1}$) along $e_j$.	It moves if and only if $e_j \in E(G_t)$; otherwise it stays put.
		More formally, for every $i \in [N]$, let $s_i(t)$ denote the state of agent $a_i$ after step $t$, and define $s_i(0) := i$. For $t \geq 1$, if $s_i(t-1) = q$, then
		\[
		s_i(t)=
		\begin{cases}
			(q \bmod N) + 1 & \text{if } e_q \in E(G_t),\\
			q   & \text{if } e_q\notin E(G_t).
		\end{cases}
		\]
		
		\item \textbf{Elimination}. Denote by $A(t)$ the set of \emph{active} agents after step $t$, and define $A(0) := A$.
		Also denote by $D_i(t)$ the set of states visited by agent $a_i$ up to time $t$; that is, $D_i(t)$ is the subset of $[N]$ corresponding to the circular interval
		$\cint{i}{s_i(t)}$, and $D_i(0) = \{i\}$.
		We say that an agent $a_i$  is \textit{redundant after step $t$} if the set $D_i(t)$ of its visited states after step $t$ is covered by the sets of visited states of the other active agents. We define $A(t)$ as a set of agents obtained from $A(t-1)$ by iteratively removing an arbitrary redundant agent until none remain. Thus, by the end of this phase we have that
		\[
			D_i(t) \not\subseteq \bigcup_{j \neq i,\; a_j \in A(t)} D_j(t),
		\]
		holds for every agent $a_i \in A(t)$.
	\end{enumerate}
	
	The following observation is a straightforward consequence of definitions.
	
	\begin{observation}\label{obs:all-states-visited}
		For every $t \in [N]$, $\bigcup_{a_i \in A(t)} D_i(t) = [N]$.
	\end{observation}
	
	We now establish a number of useful facts about the roundabout exploration process.
	
	\begin{lemma}\label{obs:diff-states}
		For every $t \in [N]$, no two active agents are in the same state after step $t$.
	\end{lemma}
	\begin{proof}
		Assume for a contradiction that after step $t\in[N]$ there exist two distinct active agents $a_i,a_j\in A(t)$ with the same state $s_i(t)=s_j(t)=q$. Then, by definition,
		\[
			D_i(t)=\cint{i}{q}
			\qquad\text{and}\qquad
			D_j(t)=\cint{j}{q}.
		\]
		Since, both sets are circular intervals that end at the same point $q$, they are nested, i.e., one contains the other. Hence one of the two agents is redundant after the movement phase at step~$t$ and would have been removed in the elimination phase, a contradiction to both agents being active after step $t$.
	\end{proof}
	
	\begin{lemma}\label{obs:1-state-3-agents}
		For every $t \in [N]$, no state has been visited by three agents that are active after step $t$, i.e., there are no distinct agents $a_i, a_j, a_k \in A(t)$ such that $D_i(t) \cap D_j(t) \cap D_k(t) \neq \emptyset$.
	\end{lemma}
	\begin{proof}
		Suppose, towards a contradiction, that there exists $q \in D_i(t) \cap D_j(t) \cap D_k(t)$. Without loss of generality, assume 
        \begin{equation}\label{eq:k-j-i}
            \cint{k}{q} \subset \cint{j}{q} \subset \cint{i}{q}.
        \end{equation}
		Since the circular intervals $\cint{i}{s_i(t)}$, $\cint{j}{s_j(t)}$, $\cint{k}{s_k(t)}$ are pairwise incomparable (due to the non-redundancy of the agents) and their common intersection contains $q$, we derive from \cref{eq:k-j-i} that 
		\[
			\cint{q}{s_i(t)} \subset \cint{q}{s_j(t)} \subset \cint{q}{s_k(t)}.
		\]
		This implies that $\cint{j}{s_j(t)} \subseteq \cint{i}{q} \cup \cint{q}{s_k(t)} \subseteq \cint{i}{s_i(t)} \cup \cint{k}{s_k(t)}$, which is not possible as $a_j$ is non-redundant after step $t$. 
	\end{proof}
	
	\begin{lemma}\label{lem:sum-of-states}
		For every $t \in [N]$, we have 
		\[
		\sum_{a_i \in A(t)} |D_i(t)| \leq 2N-|A(t)|.
		\]
	\end{lemma}
	\begin{proof}
		It follows from \cref{obs:all-states-visited} and \cref{obs:1-state-3-agents} that, by the end of step $t$, every state is visited either by one or two agents in $A(t)$. 
		Denote by $Q_1, Q_2 \subseteq [N]$ the sets of states of the first and the second kind respectively. Then
		\begin{equation}\label{eq:sum-states}
			\sum_{a_i \in A(t)} |D_i(t)|  = 
			|Q_1| + 2 \cdot |Q_2| =
			|Q_1| + 2 \cdot (N - |Q_1|) = 2N - |Q_1|.
		\end{equation}
		Observe that for every agent in $a_i \in A(t)$ there exists a state $q \in D_i(t)$ that is visited only by that agent, i.e., $q \not\in D_j(t)$ for every $a_j \in A(t)$ with $j \neq i$; indeed, otherwise $a_i$ would be redundant. This implies $|Q_1| \geq |A(t)|$, which together with \cref{eq:sum-states} gives the desired bound.
	\end{proof}

	\begin{lemma}\label{lem:agents-after-N-steps}
		Let $t \in [N]$ and $N \geq r \geq 2k+1$.
		If $r$ agents remain active after $t$ steps, i.e., $|A(t)| \geq r$, then 
		\[
		t \leq \frac{2N-r}{r-2k}.
		\] 
	\end{lemma}
	\begin{proof}
		Suppose $|A(t)| \geq r$. Then, by \cref{lem:sum-of-states}, we have
		\begin{equation}\label{eq:sum-of-states-by-time-t}
			\sum_{a_i \in A(t)} |D_i(t)| \leq 2N - r.
		\end{equation}
		On the other hand, at each of the $t$ steps, due to the $k$-edge-deficiency and \cref{obs:diff-states}, 
        %TODO (VZ): why do we need \cref{obs:diff-states} here?
        every agent in $A(t)$, except possibly $2k$ of them, makes a move and thus extends their set of visited states by one element. Thus, at every step each of the moving agents contributes 1 to the sum in \cref{eq:sum-of-states-by-time-t}, and therefore 
		\[
			t(r-2k) \leq 
            \sum_{a_i \in A(t)} |D_i(t)|.
		\] 
		Combining this with \cref{eq:sum-of-states-by-time-t}, we obtain $t(r-2k) \leq 2N - r$, which implies the desired inequality.
	\end{proof}
	
	\begin{corollary}\label{cor:active-agents}
		Let $t = \left\lfloor \frac{N}{2k} \right\rfloor$. Then, at most $6k$ agents remain active after $t$ steps, i.e., $|A(t)| \leq 6k$.
	\end{corollary}
	\begin{proof}
        If $6k > N$, then the statement is obviously true as $|A(t)|$ never exceeds $N$.
		Otherwise, towards a contradiction, suppose $|A(t)| > 6k$. Then, by \cref{lem:agents-after-N-steps} with $r=6k$,
		\[
		t \leq \frac{2N-6k}{4k} = \frac{N}{2k} - \frac{3}{2} < \left\lfloor \frac{N}{2k} \right\rfloor = t,
		\]
		a contradiction.
	\end{proof}

    The next lemma shows that each active agent is responsible for covering the interval up to the next active agent on the DFS tour. It follows from \cref{obs:all-states-visited} and the fact that the relative order of the active agents along the DFS tour does not change throughout the process.
	
	\begin{lemma}\label{cl:interval-cover}
		Let $t \in [N]$ and let $s = |A(t)|$. Let $1 \leq i_1 < i_2 \cdots < i_s \leq  N$ be the initial states of the agents in $A(t)$.
		Then, for every $\ell \in [s]$, by the end of time step $t$ agent $a_{i_{\ell}}$ visits all states between its initial state and the state preceding the initial state of the next agent, i.e., we have $\coint{i_{\ell}}{i_{\ell+1}} \subseteq D_{i_{\ell}}(t)$  for every ${\ell} \in [s-1]$, and $\coint{i_s}{i_1} \subseteq D_{i_s}(t)$.
	\end{lemma}
	\begin{proof}
		By \cref{obs:all-states-visited}, the sets $D_{i_{\ell}}(t), \ell \in [s]$ cover $[N]$. \cref{obs:diff-states} implies that during the roundabout exploration process relative order of the agents $a_{i_1}, a_{i_2}, \ldots, a_{i_s}$ along the DFS tour does not change. Furthermore, for every $\ell \in [s-1]$ no agent starts in $\coint{i_{\ell}}{i_{\ell+1}}$ except $a_{i_{\ell}}$, and no agent except $a_{i_s}$ starts in $\coint{i_s}{i_1}$. These imply that $\coint{i_{\ell}}{i_{\ell+1}} \subseteq D_{i_{\ell}}(t)$ for $\ell \in [s-1]$, and $\coint{i_s}{i_{1}} \subseteq D_{i_s}(t)$ as claimed.
	\end{proof}

	%%%%%%%%%%%%%%%%%%%%%%%%%%%%	
	\subsection{Exploration by multiple agents}\label{sec:patroling-agents}
	%%%%%%%%%%%%%%%%%%%%%%%%%%%%	

	We now consider the execution of multiple instances of the roundabout exploration process on disjoint time intervals and show that there is a choice of one agent from each execution so that the combined exploration schedules of the chosen agents cover all states of the DFS tour and thus all vertices of $\cG$. In the next section, we will turn such a distributed exploration schedule to a single-agent schedule.
	
	For the rest of the paper we fix $t := \left\lfloor \frac{N}{2k} \right\rfloor$ and $\rho := \lceil 18k \ln{(6k)} \rceil$.

    Recall that $\cG$ contains at least $\rho(\Delta+ \frac{n}{k})$ snapshots that are $k$-edge-deficient with respect to $T$.
    We partition the initial part of the timeline $[L]$ into $\rho$ time intervals, which we call \emph{epochs}, such that  within every epoch $\cG$ contains $(\Delta+ \frac{n}{k}) \geq \Delta + t$ snapshots that are $k$-edge-deficient  with respect to $T$.
    Each epoch is partitioned into two parts: the \emph{repositioning part} consisting of the first $\Delta$ time steps,
    followed by the \emph{roundabout part} consisting of the remaining time steps. By our assumption, in each roundabout part $\cG$ contains at least $t$ (not necessarily consecutive) $k$-edge-deficient snapshots.
    These snapshots are used to  execute the roundabout exploration process, assuming that
    at the beginning of the roundabout part every agent $a_i$ is positioned at state $i$.
    The repositioning part of each epoch will later be used to move the explorer to the
    appropriate starting position.

    For each $i \in [\rho]$, let $A^{(i)} \subseteq A$ be the set of agents that remained active after the $t$-th step of the roundabout process in the $i$-th epoch, and denote by $S^{(i)}$ the set of initial states of these agents. By \cref{cor:active-agents}, $|A^{(i)}| = |S^{(i)}| \leq 6k$.
    Let
    \[
        M = \bigcup_{i \in [\rho]} S^{(i)} = \{ m_1, m_2, \ldots, m_d \},
    \]
    where $1 \leq m_1 < m_2 < \cdots < m_d \leq N$, and note that $d \leq  6k\rho$.
    Denote $I_i = \coint{m_i}{m_{i+1}}$,
    for $i \in [d-1]$, $I_d = \coint{m_d}{m_{1}}$, and $\bfI = \{ I_1, I_2, \ldots, I_d \}$.
    
    \begin{lemma}\label{lem:visited-interval}
        For each $i \in [\rho]$ and $I \in \bfI$,  there exists an agent in $A^{(i)}$ that visits the entire interval $I$ during the $i$-th epoch.
    \end{lemma}
    \begin{proof}
        Since the intervals in $\bfI$ are obtained by taking the common refinement of the families of intervals induced by the sets $S^{(i)}$ for $i \in [\rho]$, it follows that for every $i \in [\rho]$ and every interval $I \in \bfI$, there exist two consecutive states $p, q \in S^{(i)}$ such that $I \subseteq \coint{p}{q}$. Therefore, by \cref{cl:interval-cover}, the agent $a_p$ visits all states in $I$ during the $i$-th epoch.
    \end{proof}

    Let $I \in \bfI$. We say that a tuple $(s_1, s_2, \ldots, s_{\rho}) \in S^{(1)} \times S^{(2)} \times \cdots \times S^{(\rho)}$ is \emph{$I$-missing} if for every $i \in [\rho]$, the agent in $A^{(i)}$ with the initial state $s_i$ does not cover $I$ entirely in the $i$-th epoch. In other words, none of the agents with starting states $s_1, s_2, \ldots, s_{\rho}$ in epochs $1,2, \ldots, \rho$, respectively, covers $I$ in their epoch. We say that $(s_1, s_2, \ldots, s_{\rho})$ is \emph{$I$-covering}, if it is not $I$-missing. 
    
\begin{lemma}\label{lem:I-covering}
   At least $\frac{1}{12}$ of the tuples from $S^{(1)} \times S^{(2)} \times \cdots \times S^{(\rho)}$ are $I$-covering for every $I \in \bfI$.
\end{lemma}
\begin{proof}
    Fix an interval $I \in \bfI$, and let $F_I$ be the set of tuples that are
    $I$-missing. By \cref{lem:visited-interval}, for each $i \in [\rho]$ there is
    at least one element in $S^{(i)}$ whose corresponding agent visits the whole
    interval $I$ during the $i$-th epoch. Hence, in order for a tuple to be
    $I$-missing, its $i$-th coordinate must avoid at least one element of
    $S^{(i)}$, for every $i \in [\rho]$. Therefore
    \[
        |F_I|
        \leq
        \prod_{i \in [\rho]} \bigl(|S^{(i)}|-1\bigr).
    \]
    It follows that the number of tuples that are $I$-missing
    for at least one interval $I \in \bfI$ is at most
    \[
        \sum_{I \in \bfI} |F_I|
        \leq
        d \cdot \prod_{i \in [\rho]} \bigl(|S^{(i)}|-1\bigr).
    \]
    Thus, the proportion of tuples that fail to be $I$-covering for at least one
    $I \in \bfI$ is at most
    \[
        d \cdot \prod_{i \in [\rho]} \bigl(|S^{(i)}|-1\bigr)
        \Big/
        \prod_{i \in [\rho]} |S^{(i)}|
    \leq
        d \cdot
        \prod_{i \in [\rho]}
        \left(1-\frac{1}{|S^{(i)}|}\right).
    \]
    We now show that this quantity is at most $11/12$. Using
    $|S^{(i)}| \leq 6k$ for every $i \in [\rho]$, we have
    \[
        \prod_{i \in [\rho]}
        \left(1-\frac{1}{|S^{(i)}|}\right)
        \leq
        \left(1-\frac{1}{6k}\right)^\rho
        \leq
        e^{-\rho/6k}.
    \]
    Since $d \leq 6k\rho$ and $\rho=\lceil 18k \ln(6k)\rceil$, this gives
    \begin{align*}
        d \cdot
        \prod_{i \in [\rho]}
        \left(1-\frac{1}{|S^{(i)}|}\right)
        &\leq
        6k\rho \cdot e^{-\rho/6k} \\
        &\leq
        6k \cdot \lceil 18k \ln(6k)\rceil
        \cdot e^{-3\ln(6k)} \\
        &=
        \frac{\lceil 18k \ln(6k)\rceil}{(6k)^2} \\
        &\leq
        \frac{\lceil 18\ln 6\rceil}{6^2}
        =
        \frac{11}{12}.
    \end{align*}
    Consequently, at least a $1/12$ fraction of all tuples are $I$-covering for every $I \in \bfI$, as required.
\end{proof}

 	%%%%%%%%%%%%%%%%%%%%%%%%%%%%	
	 \subsection{Single-agent exploration}\label{sec:single-agent-exploration}
	%%%%%%%%%%%%%%%%%%%%%%%%%%%%	
 	
 	We now show how to construct a single-agent exploration schedule within the $\rho$ epochs from 
    a tuple that is  $I$-covering for every $I \in \bfI$.
 	
 	\begin{lemma}\label{lem:tuple-to-exploration}
 		Let $(s_1, s_2, \ldots, s_{\rho}) \in S^{(1)} \times S^{(2)} \times \cdots \times S^{(\rho)}$ be a tuple that is $I$-covering for every $I \in \bfI$. Then $\cG$ can be explored from any vertex during the $\rho$ epochs.
 	\end{lemma}
	\begin{proof}
        For each $j \in [\rho]$, let $a^{(j)}$ be the agent that starts from state $s_j$ in
        the roundabout exploration process of the $j$-th epoch, and let $D^{(j)}$ denote the
        set of states visited by this agent during the roundabout part of that epoch.
        By the definition of $I$-covering, for every interval $I \in \bfI$ there
        exists some $j \in [\rho]$ such that $I \subseteq D^{(j)}$.
        Since the intervals in $\bfI$ form a partition of $[N]$, it follows that
        \[
        \bigcup_{j=1}^{\rho} D^{(j)} = [N].
        \]
        
        We now describe a single-agent exploration schedule that simulates the agents
        $a^{(1)},\dots,a^{(\rho)}$ across the $\rho$ epochs.
        Recall that each epoch consists of a repositioning part of $\Delta$ steps, followed by a roundabout part in which $\cG$ contains at least $t$ snapshots that are $k$-edge-deficient.
        At the beginning of epoch $j$, the explorer uses the $\Delta$-temporal connectivity
        of $\cG$ to move, during the repositioning part, to the vertex corresponding to state
        $s_j$.
        During the roundabout part, the explorer moves only within $k$-edge-deficient snapshots replicating exactly the
        moves performed by agent $a^{(j)}$ in the roundabout exploration process.
        By construction, during the roundabout part of epoch $j$ the explorer visits
        precisely the states in $D^{(j)}$.
        Therefore, over all $\rho$ epochs, the explorer visits all states in
        $\bigcup_{j=1}^{\rho} D^{(j)} = [N]$, and hence visits all vertices of $\cG$.
	\end{proof}
 
    We now have everything ready to prove \cref{th:k-def-exploration}, which we restate below for convenience.
     
    \KDefExploration*
    \begin{proof} 
        Recall that $N=2(n-1)$, $t = \left\lfloor \frac{N}{2k} \right\rfloor$, and $\rho = \lceil 18k \ln(6k) \rceil$.
        Let $L$ be the lifetime of $\cG$ and assume $\cG$ contains at least $\rho(\Delta+ \frac{n}{k})$ snapshots that are $k$-edge-deficient with respect to $T$. Then the timeline $[L]$ can be partitioned into $\rho$ intervals such that, in each interval, $\cG$ contains at least $\Delta + \frac{n}{k} \geq \Delta + t$ snapshots, each of which is $k$-edge-deficient. 
		It follows from \cref{lem:I-covering} and \cref{lem:tuple-to-exploration} that $\cG$ can be explored from any  vertex.
    \end{proof}

    %%%%%%%%%%%%%%%%%%%%%%%%%%%%%%%%%%%%%%
    \section{Computing exploration schedules efficiently}
    \label{sec:efficient-computation}
    %%%%%%%%%%%%%%%%%%%%%%%%%%%%%%%%%%%%%%

    In this section we prove \cref{th:k-def-exploration-alg}, which makes the exploration bound of \cref{th:k-def-exploration} (up to a constant factor) algorithmic for $k$-edge-deficient temporal graphs. The main algorithmic difficulty of \cref{th:k-def-exploration-alg} is that the spanning tree witnessing $k$-edge-deficiency is not part of the input.
    To address this, we first show in \cref{lem:known-span-exploration} how to turn the constructive proof of \cref{th:k-def-exploration} into an algorithm for computing an exploration schedule, assuming the spanning tree $T$ is given.
    We then show in \cref{lem:find-good-tree} how to find a suitable tree efficiently in a given $k$-edge-deficient temporal graph such that sufficiently many snapshots are $2k$-edge-deficient with respect to that tree. Finally, by combining the two lemmas, we derive \cref{th:k-def-exploration-alg}.

    Throughout this section we write
    $\tau(n,k,\Delta) := \lceil 18k \ln(6k) \rceil \left(\Delta+\frac{n}{k}\right)$.

    %%%%%%%%%%%%%%%%%%%%%%%%%%%%%%%%%%%%%%
    \subsection{Computing a schedule when the spanning tree is given}
    \label{subsec:efficient-given-tree}
    %%%%%%%%%%%%%%%%%%%%%%%%%%%%%%%%%%%%%%

    We first assume that the spanning tree $T$ from \cref{th:k-def-exploration} is given as part of the input.

    \begin{lemma}\label{lem:known-span-exploration}
        Let $\Delta$, $k$, and $n$ be natural numbers.
        Let $\cG$ be a $\Delta$-temporally connected temporal graph on $n$ vertices, and let $T$ be a spanning tree of the underlying graph of $\cG$. Suppose that $\cG$ contains at least $\tau(n,k,\Delta)$ snapshots that are $k$-edge-deficient with respect to $T$.
        Given $\cG$, $T$, and a vertex $v$, an exploration schedule of $\cG$ from $v$ can be computed by a Las Vegas algorithm in expected polynomial time. Moreover, if  $k$ is a constant, then such a schedule can be computed deterministically in polynomial time.
    \end{lemma}

    \begin{proof}
        We follow the constructive proof of \cref{th:k-def-exploration}. Recall that
        \[
            N=2(n-1), \qquad
            t=\left\lfloor \frac{N}{2k}\right\rfloor,
            \qquad
            \rho=\lceil 18k\ln(6k)\rceil .
        \]
        First, we identify the snapshots that are $k$-edge-deficient with respect to $T$. For a snapshot $G_j$, this amounts to checking whether
        $|E(T)\setminus E(G_j)|\leq k$,
        which can be done in polynomial time.

        We then scan the timeline from left to right and partition an initial part of it into $\rho$ consecutive epochs, each containing at least
        $\Delta+\frac{n}{k}$
        snapshots that are $k$-edge-deficient with respect to $T$. As in the proof of \cref{th:k-def-exploration}, the first $\Delta$ time steps of each epoch are used for repositioning, and the remaining part contains at least $t$ snapshots that are $k$-edge-deficient.

        For each epoch $i \in [\rho]$, we simulate the roundabout exploration process on the first $t$ $k$-edge-deficient snapshots in its roundabout part. The movement phase is straightforward to simulate. The elimination phase can also be implemented in polynomial time: for each active agent, we check whether the circular interval it has visited is contained in the union of the intervals visited by the other active agents, and remove it if this is the case. Repeating this until no redundant agent remains yields the set $A^{(i)}$ of active agents and the corresponding set $S^{(i)}$ of their initial states. By \cref{cor:active-agents}, we have
        $|S^{(i)}|\leq 6k$ for every $i \in [\rho]$.

        We next construct the common refinement $\bfI$ of the interval partitions induced by the sets $S^{(i)}$, exactly as in \cref{sec:patroling-agents}. The number of intervals in $\bfI$ is at most $6k\rho$, and clearly the construction is polynomial.

        It remains to choose a tuple
        $(s_1,\ldots,s_\rho) \in
            S^{(1)}\times\cdots\times S^{(\rho)}$
        that is $I$-covering for every $I \in \bfI$. By \cref{lem:I-covering}, at least a $1/12$ fraction of all tuples have this property. Hence, if we sample a tuple uniformly at random and check whether it covers every interval in $\bfI$, the expected number of trials before success is at most 12. For a fixed tuple, the check is polynomial: we compute the union of the intervals covered by the selected agents and test whether it contains every interval of $\bfI$, equivalently all states of the DFS tour.

        Once a covering tuple has been found, the single-agent exploration schedule is obtained as in \cref{lem:tuple-to-exploration}. During the repositioning part of each epoch, we compute a temporal walk of length at most $\Delta$ from the explorer's current vertex to the vertex corresponding to the selected state $s_i$.
        Such a walk exists by $\Delta$-temporal connectivity, and it can be found efficiently by the standard algorithm for foremost paths in temporal graphs \cite{XFJ03}. During the roundabout part, the explorer follows the movements of the selected virtual agent, waiting during snapshots that are not used by the roundabout process.

        This gives a Las Vegas algorithm with expected polynomial running time. If $k$ is a constant, then $\rho=O(1)$ and $|S^{(i)}|\leq 6k=O(1)$, so the total number of tuples
        \[
            \prod_{i\in[\rho]} |S^{(i)}|
            \leq
            (6k)^\rho
        \]
        is constant. In this case we can deterministically enumerate all tuples and find a covering one in polynomial time.
    \end{proof}

    %%%%%%%%%%%%%%%%%%%%%%%%%%%%%%%%%%%%%%
    \subsection{Computing a suitable spanning tree}
    \label{subsec:efficient-finding-tree}
    %%%%%%%%%%%%%%%%%%%%%%%%%%%%%%%%%%%%%%

    We now consider the setting in which the input temporal graph $\cG$ is $k$-edge-deficient, but the spanning tree witnessing this property is not given. We first show in \cref{lem:find-good-tree} that, by looking at a sufficiently long prefix of the timeline of the temporal graph, we can compute a spanning tree with respect to which half of the snapshots are $2k$-edge-deficient. We then combine this with \cref{lem:known-span-exploration} to derive \cref{th:k-def-exploration-alg}.

    \begin{lemma}\label{lem:find-good-tree}
        Let $\cG=\langle G_1,G_2,\dots,G_L\rangle$ be a $k$-edge-deficient temporal graph with underlying graph $G$, and let $q$ be a positive integer such that $2q\leq L$. Then one can compute, in polynomial time, a spanning tree $T$ of $G$ such that among the first $2q$ snapshots of $\cG$, at least $q$ snapshots are $2k$-edge-deficient with respect to $T$.
    \end{lemma}

    \begin{proof}
        Since $\cG$ is $k$-edge-deficient, there exists a spanning tree $T_0$ of $G$ such that every snapshot of $\cG$ is $k$-edge-deficient with respect to $T_0$. We do not know $T_0$, but we can find a suitable substitute as follows.

        For every edge $e\in E(G)$, define
        $w(e) := \bigl|\{i\in[2q] : e\notin E(G_i)\}\bigr|$.
        That is $w(e)$ is the number of snapshots among the first $2q$ snapshots in which $e$ is missing. Let $T$ be a minimum-weight spanning tree of $G$ with respect to the weights $w$. Such a tree can be computed in polynomial time.

        For a spanning tree $S$ of $G$ and a snapshot $G_i$, write
        $w(G_i,S):=|E(S)\setminus E(G_i)|$
        for the number of edges of $S$ missing from $G_i$. By the choice of $T$, we have
        \[
            \sum_{i\in[2q]} w(G_i,T)
            =
            \sum_{e\in E(T)} w(e)
            \leq
            \sum_{e\in E(T_0)} w(e)
            =
            \sum_{i\in[2q]} w(G_i,T_0)
            \leq
            2qk.
        \]
        Hence at most $q$ of the first $2q$ snapshots can satisfy
        $w(G_i,T)>2k$,
        since otherwise the above sum would exceed $2qk$. Therefore, at least $q$ of the first $2q$ snapshots satisfy $w(G_i,T)\leq 2k$, meaning that they are $2k$-edge-deficient with respect to $T$.
    \end{proof}

    We now have everything to prove \cref{th:k-def-exploration-alg}, which we restate for convenience.

    \KDefExplorationAlg*
    \begin{proof}
        Let $q := \tau(n,2k,\Delta) = \lceil 36k \ln{(12k)} \rceil(\Delta+ \frac{n}{2k})$.
        Since the lifetime of $\cG$ is at least $2q$, we may apply \cref{lem:find-good-tree} to the first $2q$ snapshots. This gives, in polynomial time, a spanning tree $T$ of the underlying graph of $\cG$ such that at least $q$ of these snapshots are $2k$-edge-deficient with respect to $T$.

        Since the temporal graph $\cG_{[1,2q]}$ is $\Delta$-temporally connected, by \cref{lem:known-span-exploration}, applied with $2k$ in place of $k$, it can be explored from any vertex and the corresponding exploration schedule can be computed by a Las Vegas algorithm in expected polynomial time or deterministically in polynomial time if $k$ is a constant.
    \end{proof}

    %%%%%%%%%%%%%%%%%%%%%%%%%%%%%%%%%%%%%%
    \section{Conclusion}
    %%%%%%%%%%%%%%%%%%%%%%%%%%%%%%%%%%%%%%
    
    We have shown that temporal exploration in always-connected $k$-edge-deficient temporal graphs can be
    performed in time linear in the number of vertices, up to a factor depending only on $k$.
    More precisely, every always-connected $k$-edge-deficient temporal graph on $n$ vertices
    admits an exploration schedule of length $O(nk\log k)$. This removes the logarithmic
    dependence on $n$ from the previous $O(kn\log n)$ upper bound and nearly matches the
    known $\Omega(n\log k)$ lower bound \cite{ES22}. In particular, when $k$ is constant, the exploration
    time is $\Theta(n)$, showing that this near-static regime retains the linear-time behaviour
    of static graph traversal.
    
    Our results also give a more general view of the parameters governing temporal exploration.
    The bound extends beyond always-connected temporal graphs to the setting of
    $\Delta$-temporal connectivity, where the exploration time depends linearly on the temporal
    connectivity parameter $\Delta$. Thus, the complexity of exploration is controlled not only
    by the amount of temporal disruption, captured by $k$, but also by the robustness of
    connectivity over time, captured by $\Delta$.
    
    The bound is constructive. For $k$-edge-deficient temporal graphs, the corresponding exploration schedule can be computed by a Las Vegas algorithm in expected polynomial time, and deterministically in polynomial time when  $k$ constant.
    
    Several questions remain open. The most immediate one is whether the remaining factor of
    $k$ in the gap between the upper bound $O(nk\log k)$ and the lower bound
    $\Omega(n\log k)$ can be removed. More broadly, it would be interesting to understand
    whether similar linear-in-$n$ guarantees can be obtained under weaker notions of temporal
    stability. Such results would further clarify which aspects of temporal variation are
    responsible for the gap between static and temporal exploration.

	%%% Bibliography
	\bibliographystyle{alpha}
    \bibliography{references}

@article{CFQS12,
  author  = {Casteigts, Arnaud and Flocchini, Paola and Quattrociocchi, Walter and Santoro, Nicola},
  title   = {Time-varying graphs and dynamic networks},
  journal = {International Journal of Parallel, Emergent and Distributed Systems},
  volume  = {27},
  number  = {5},
  pages   = {387--408},
  year    = {2012},
  doi     = {10.1080/17445760.2012.668546},
  url     = {https://doi.org/10.1080/17445760.2012.668546}
}

@article{MichailSurvey,
  author  = {Michail, Othon},
  title   = {An Introduction to Temporal Graphs: An Algorithmic Perspective},
  journal = {Internet Mathematics},
  volume  = {12},
  number  = {4},
  pages   = {239--280},
  year    = {2016},
  doi     = {10.1080/15427951.2016.1177801},
  url     = {https://doi.org/10.1080/15427951.2016.1177801}
}

@article{MS16,
  author  = {Michail, Othon and Spirakis, Paul G.},
  title   = {Traveling salesman problems in temporal graphs},
  journal = {Theoretical Computer Science},
  volume  = {634},
  pages   = {1--23},
  year    = {2016},
  doi     = {10.1016/j.tcs.2016.04.006},
  url     = {https://doi.org/10.1016/j.tcs.2016.04.006}
}

@article{EHK21,
  author  = {Erlebach, Thomas and Hoffmann, Michael and Kammer, Frank},
  title   = {On temporal graph exploration},
  journal = {Journal of Computer and System Sciences},
  volume  = {119},
  pages   = {1--18},
  year    = {2021},
  doi     = {10.1016/j.jcss.2021.01.005},
  url     = {https://doi.org/10.1016/j.jcss.2021.01.005}
}

@article{BZ19,
  author  = {Bodlaender, Hans L. and van der Zanden, Tom C.},
  title   = {On exploring always-connected temporal graphs of small pathwidth},
  journal = {Information Processing Letters},
  volume  = {142},
  pages   = {68--71},
  year    = {2019},
  doi     = {10.1016/j.ipl.2018.10.016},
  url     = {https://doi.org/10.1016/j.ipl.2018.10.016}
}

@inproceedings{ES18,
  title={Faster Exploration of Degree-Bounded Temporal Graphs},
  author={Erlebach, Thomas and Spooner, Jakob T},
  booktitle={43rd International Symposium on Mathematical Foundations of Computer Science (MFCS 2018)},
  year={2018}
}

@inproceedings{EKLSS19,
  title={Two moves per time step make a difference},
  author={Erlebach, Thomas and Kammer, Frank and Luo, Kelin and Sajenko, Andrej and Spooner, Jakob T},
  booktitle={46th International Colloquium on Automata, Languages, and Programming (ICALP 2019)},
  pages={141},
  year={2019}
}

@inproceedings{AGMZ22,
  title={Faster exploration of some temporal graphs},
  author={Adamson, Duncan and Gusev, Vladimir V and Malyshev, Dmitriy and Zamaraev, Viktor},
  booktitle={1st Symposium on Algorithmic Foundations of Dynamic Networks (SAND 2022)},
  pages={5--1},
  year={2022}
}

@article{ES22,
  author  = {Erlebach, Thomas and Spooner, Jakob T.},
  title   = {Exploration of $k$-Edge-Deficient Temporal Graphs},
  journal = {Acta Informatica},
  volume  = {59},
  number  = {4},
  pages   = {387--407},
  year    = {2022},
  doi     = {10.1007/s00236-022-00421-5},
  url     = {https://doi.org/10.1007/s00236-022-00421-5}
}

@article{ES23,
  title={Parameterised temporal exploration problems},
  author={Erlebach, Thomas and Spooner, Jakob T},
  journal={Journal of Computer and System Sciences},
  volume={135},
  pages={73--88},
  year={2023},
  publisher={Elsevier}
}

@inproceedings{AFGW23,
  title={Kernelizing Temporal Exploration Problems},
  author={Arrighi, Emmanuel and Fomin, Fedor V and Golovach, Petr A and Wolf, Petra},
  booktitle={18th International Symposium on Parameterized and Exact Computation (IPEC 2023)},
  year={2023}
}

@inproceedings{DEKMM24,
  title={Exploiting Automorphisms of Temporal Graphs for Fast Exploration and Rendezvous},
  author={Dogeas, Konstantinos and Erlebach, Thomas and Kammer, Frank and Meintrup, Johannes and Moses Jr, William K},
  booktitle={51st International Colloquium on Automata, Languages, and Programming (ICALP 2024)},
  pages={55--1},
  year={2024}
}

@article{BGMR25,
  title={Improved exploration of temporal graphs},
  author={Bastide, Paul and Groenland, Carla and Michel, Lukas and Rambaud, Cl{\'e}ment},
  journal={arXiv preprint arXiv:2511.22604},
  year={2025}
}

@article{HS12,
  title={Temporal networks},
  author={Holme, Petter and Saram{\"a}ki, Jari},
  journal={Physics reports},
  volume={519},
  number={3},
  pages={97--125},
  year={2012},
  publisher={Elsevier}
}

@book{HS19,
  title={Temporal network theory},
  author={Holme, Petter and Saram{\"a}ki, Jari},
  volume={2},
  year={2019},
  publisher={Springer}
}

@inproceedings{BGKSSZ26,
  title={Temporal exploration of random spanning tree models},
  author={Baguley, Samuel and G{\"o}bel, Andreas and Klodt, Nicolas and Skretas, George and Sylvester, John and Zamaraev, Viktor},
  booktitle={Proceedings of the 2026 Annual ACM-SIAM Symposium on Discrete Algorithms (SODA 2026)},
  pages={2876--2887},
  year={2026}
}

@mastersthesis{Tag20,
  title        = {Exploring Temporal Cycles and Grids},
  author       = {Taghian Alamouti, Shadi},
  year         = {2020},
  school       = {Concordia University},
  type         = {MSc thesis}
}

@inproceedings{IKW14,
  title={Exploration of constantly connected dynamic graphs based on cactuses},
  author={Ilcinkas, David and Klasing, Ralf and Wade, Ahmed Mouhamadou},
  booktitle={21st International Colloquium on Structural Information and Communication Complexity (SIROCCO 2014)},
  pages={250--262},
  year={2014},
  organization={Springer}
}

@article{XFJ03,
  title={Computing shortest, fastest, and foremost journeys in dynamic networks},
  author={Xuan, B Bui and Ferreira, Afonso and Jarry, Aubin},
  journal={International Journal of Foundations of Computer Science},
  volume={14},
  number={02},
  pages={267--285},
  year={2003},
  publisher={World Scientific}
}
\end{document}